\documentclass[letterpaper, onecolumn, 10 pt]{ieeeconf}
\IEEEoverridecommandlockouts

\usepackage{amsmath, amsthm, amssymb, amsfonts, mathtools, xargs, tensor, units, cite, stmaryrd, mathrsfs, algpseudocode, comment, graphicx}
\usepackage{soul} %for strikeouts during commenting
\usepackage{mathrsfs}
\usepackage{multibib}
\usepackage{float}
\usepackage{setspace}
\usepackage[svgnames]{xcolor}
\usepackage[unicode=true, bookmarks=true, bookmarksnumbered=false, bookmarksopen=false, breaklinks=false, pdfborder={0 0 1}, backref=false, colorlinks=false, hidelinks]{hyperref}

\makeatletter
\DeclareMathOperator\sign{sgn}
\DeclareMathOperator\diag{diag}
\let\NAT@parse\undefined
\makeatother

\theoremstyle{definition}

\theoremstyle{definition}
\newtheorem{assum}{Assumption}

\newtheorem{theorem}{Theorem}

\theoremstyle{remark}
\newtheorem{remark}{Remark}

\allowdisplaybreaks
\usepackage{geometry}
\title{Output Feedback Adaptive Optimal Control of Affine Nonlinear systems with a Linear Measurement Model}
\author{\normalsize Tochukwu Elijah Ogri$^{1}$ \and S M Nahid Mahmud$^{2}$ \and Zachary I. Bell$^{3}$ \and Rushikesh Kamalapurkar$^{1}$% <-this % stops a space
\thanks{*This research was supported in part by Office of Naval Research under award number
N00014-21-1-2481 and the Air Force Research Laboratories under contract number FA8651-19-2-0009. Any opinions, findings, or recommendations in this article are those of the author(s), and do not necessarily reflect the views of the sponsoring agencies.}% 
\thanks{$^{1}$School of Mechanical and Aerospace Engineering, Oklahoma State University, email: {\tt\footnotesize \{tochukwu.ogri, rushikesh.kamalapurkar\} @okstate.edu}.}
\thanks{$^{2}$ School of Aeronautics and Astronautics, Purdue University, West Lafayette, 47907, USA, e-mail:  {\tt\footnotesize \{mahmud7\}@purdue.edu}.}%
\thanks{$^{3}$ Air Force Research Laboratories, Florida, USA, email: {
\tt \footnotesize zachary.bell.10@us.af.mil.}}}

\begin{document}
\maketitle
\thispagestyle{empty}
\pagestyle{empty}
\begin{abstract} 
Real-world control applications in complex and uncertain environments require adaptability to handle model uncertainties and robustness against disturbances. This paper presents an online, output-feedback, critic-only, model-based reinforcement learning architecture that simultaneously learns and implements an optimal controller while maintaining stability during the learning phase. Using multiplier matrices, a convenient way to search for observer gains is designed  along with a controller that learns from simulated experience to ensure stability and convergence of trajectories of the closed-loop system to a neighborhood of the origin. Local uniform ultimate boundedness of the trajectories is established using a Lyapunov-based analysis and demonstrated through simulation results, under mild excitation conditions.
\end{abstract}

% \begin{IEEEkeywords}
% reinforcement learning, observer design, state estimation, optimal control
% \end{IEEEkeywords}

\section{Introduction}
% Talk about reinforcement learning and state estimation using output feedback

 Reinforcement learning (RL) has proven to be robust to modeling errors in dynamic systems and capable of addressing parametric uncertainties, ensuring a fast convergence rate to the optimal solution while maintaining stability regardless of disturbances in the system \cite{SCC.Kamalapurkar2014,SCC.Mahmud.Nivison.ea2021, SCC.Kamalapurkar.Rosenfeld.ea2016}. 

Real-time state estimation while maintaining system stability is offered by model-based reinforcement learning (MBRL)  through the use of neural networks (NNs) to achieve fast approximation without complete knowledge of the dynamics of the system\cite{SCC.Kamalapurkar2017a,SCC.Self.Harlan.ea2019}. In the absence of full state measurement information, MBRL controllers in \cite{SCC.Cichosz1999, SCC.Wawrzynski2009, SCC.Zhang.Cui.ea2011, SCC.Adam.Busoniu.ea2012}, tend to perform poorly since the excitation conditions require the accuracy
of the estimated model to guarantee the closed-loop stability of the system. To address this, the semi-definite programming (SDP) observer design technique in \cite{SCC.Wang.Yang.ea2017} is augmented to  provide an accurate estimated model for the MBRL controller which will yield a more robust framework to model uncertainties/disturbances and achieve an optimal solution while guaranteeing the stability of the closed-loop system.
 % to search for the extended Luenberger observer gains for nonlinear systems 
% General context about the methods used and existing literature

In this paper, an observer for state estimation using semi-definite programming (SDP) to search for the extended Luenberger observer gains is developed for continuous-time nonlinear systems. The observer developed in this paper uses the multiplier matrix approach to develop sufficient conditions for stability of the nonlinear control affine system with partially constrained input. By placing bounds on the derivatives of the drift and control effectiveness functions of the system, sufficient conditions developed using Lyapunov analysis are used to guarantee the stability of the state estimation error dynamics \cite{SCC.Behcet.Martin.ea2008, SCC.Quintana.Bernal.ea2021}. The state estimates are then used in a model-based reinforcement learning (MBRL) framework to design a controller that ensures the stability of the closed-loop system during learning.

% Get specific about the methods used
This observer architecture is motivated by the observer design technique developed in \cite{SCC.Wang.Yang.ea2017} which introduces a third observer gain to cancel nonconvex terms in the semidefinite condition. When compared with the extended Luenberger observer in \cite{SCC.ARCAK.ea2001}, which just has a linear correction term and a nonlinear injection term, the method developed in \cite{SCC.Wang.Yang.ea2017} extends to a less conservative class of nonlinear systems while ensuring uniform asymptotic convergence. This paper offers a modification to the observer structure in \cite{SCC.Wang.Yang.ea2017, SCC.Jung.Hee2008} with fewer restrictions on the class of nonlinear systems. The observer is then combined with an MBRL-based controller to optimize a given performance objective. 

The goal of the MBRL is to learn an optimal controller to approximate the value function, and subsequently, the optimal policy, for an input-constrained nonlinear system. While adaptive optimal control methods have been extensively studied in the literature to solve the online optimal control problem, \cite{SCC.Modares.Lewis2013, SCC.Modares.Lewis.ea2014, SCC.Yang.Liu.Huang.ea2013, SCC.Mahmud.Nivison.ea2021, arXivSCC.Mahmud.Abudia.ea2021, SCC.Cichosz1999, SCC.Wawrzynski2009, SCC.Zhang.Cui.ea2011, SCC.Adam.Busoniu.ea2012,SCC.Abu-Khalaf.Lewis2005}, most existing results require full state feedback. In this paper, an output feedback problem is solved for systems with a linear measurement model. Furthermore, unlike actor-critic MBRL methods popular in the literature \cite{SCC.Kamalapurkar.Rosenfeld.ea2016, SCC.Kamalapurkar.Klotz.ea2018}, this paper presents a critic-only structure to provide an approximate solution of the Hamilton–Jacobi–Bellman (HJB) equation that requires the identification of fewer free parameters. Lyapunov methods are used to show that the states of the system, the state estimation error, and the critic weights are locally uniformly ultimately bounded (UUB) for all time starting from any initial condition.
% Take about differences and contributions to the existing literature

  The main contributions are as follows:
  \begin{enumerate}
     \item  This paper uses an observer with bounded Jacobian for fast state estimation and an online RL critic-only architecture to learn a controller that keeps the input-constrained nonlinear systems stable during the learning phase. This novel architecture is different from existing NN network observers in \cite{SCC.Kim.Lewis.ea1997,SCC.Abdollahi.Patel2006,SCC.Yang.Liu.Huang.ea2013,SCC.Yang.Liu.ea2014,SCC.Huang.Jiang2015, SCC.farza.ea2010} whose convergence analysis relies solely on negative terms that result from a $\sigma-$modification-like term added to the weight update laws, and as a result, similar to adaptive control, the convergence of the observer weights to their true values cannot be expected, and convergence of state estimates to the true states is not robust to disturbances and approximation errors.
    \item This paper proposes a robust output feedback RL method for a nonlinear control-affine system with a general $C$ matrix. The method in this paper does not require restrictions on the form and rank of the $C$ matrix which is different from most NN  observers in the literature \cite{SCC.Yang.Liu.Huang.ea2013,SCC.Yang.Liu.ea2014,SCC.Huang.Jiang2015}. A drawback of existing state feedback control methods, like\cite{SCC.Yang.Liu.ea2014}, is that the substitution $x = C^+y$ implicitly restricts the proof to systems where the number of outputs is larger than the number of states, which invalidates the point of output feedback control.
\end{enumerate}
     	         
 The rest of the paper is organized as follows: Section ~\ref{section:problemFormulation} contains the problem formulation, Section ~\ref{section:stateEstimator} introduces the  state estimator/observer, Section ~\ref{section:sectorFormulation} presents the Multiplier matrices and sector Conditions, Section ~\ref{section:controlDesign} contains control design using MBRL methods, Section ~\ref{section:stabilityAnalysis} contains stability analysis of the developed architecture, Section ~\ref{section:simulation} contains simulation results and Section ~\ref{section:conclusion} concludes the paper.

\section{Problem Formulation}	
\label{section:problemFormulation}
This paper considers nonlinear dynamical systems of the form, 
\begin{equation}
\dot{x} = f(x)+g(x)u , \quad y = Cx, \label{eq:dynamics_x}
\end{equation}
where $x \in \mathbb{R}^{n}$ is the system state , $u \in\mathbb{R}^{m}$ is the control input,  $C \in \mathbb{R}^{q \times n}$ is the output matrix, and $y \in \mathbb{R}^{q}$ is the measured output. The functions $f: \mathbb{R}^{n} \rightarrow \mathbb{R}^{n}$, and $g: \mathbb{R}^{n} \rightarrow \mathbb{R}^{n \times m}$, denote the drift and the control effectiveness matrix, respectively.

\begin{assum}\label{ass:jacobianbounds}
    The functions $f$ and  $g$ are known, their derivatives exist on a compact set $\mathcal{C}\subset\mathbb{R}^n$, and satisfy the element-wise bounds
    \begin{gather}
	(M_{f_{1}})_{i,j} \leq \frac{\mathrm{d}(f(x))_i}{\mathrm{d}(x)_j} \leq (M_{f_{2}})_{i,j}, \\ 
	(M_{g_{1}})_{i,j,k} \leq \frac{\mathrm{d}(g(x))_{i,k}}{\mathrm{d}(x)_j} \leq (M_{g_{2}})_{i,j,k},
\end{gather}
for all $x\in\mathcal{C}$, $i,j = 1,\ldots,n$ and $k=1,\ldots,m$, where $(\cdot)_{i,j,k}$ and $(\cdot)_{i,j}$, and $(\cdot)_{i}$ denote the element of the array $(\cdot)$ at the indices indicated by the subscript. 
\end{assum}
\begin{remark}
The conditions stated in Assumption~\ref{ass:jacobianbounds} are commonly required in several observer design schemes (see, e.g., \cite{SCC.Zemouche.Bara.ea2005,SCC.Wang.Bevly.ea2014,SCC.Wang.Yang.ea2017,SCC.Rajamani.Zemouche.ea2020}.
\end{remark}
The objective is to design a uniformly asymptotic observer to estimate the states online, using input-output measurements, and to simultaneously synthesize and utilize an estimate of a controller that minimizes the cost functional defined in (\ref{costfunction}), under the saturation constraint $|{(u)_i}| \leq \bar{\lambda} > 0 \text{ for } i = 1, \hdots, m$ while ensuring local uniform ultimate boundedness of the trajectories of the system in (\ref{eq:dynamics_x}). 

\section{State Estimator}
\label{section:stateEstimator}

In this section, a state estimator inspired by the extended Luenberger observer is developed to generate estimates of $x$. 
Let the nonlinear dynamics in (\ref{eq:dynamics_x}) be expressed in the form below,
\begin{equation}
    \dot{x} =	M_{f_{1}}x+M_{g_{1}}u x+\bar{f}(x) + \bar{g}_{u}(x,u),\label{aug_dynamics_x}
\end{equation}
where
\begin{gather}
    \bar{f}(x) = -M_{f_{1}}x+f(x), \text{ and } \\
    \bar{g}_{u}(x,u) = -M_{g_{1}}ux + \sum_{i=1}^{m} g_i(x)(u)_i.
\end{gather}
The derivatives of $\bar{f}$ and $\bar{g}$ satisfy the element-wise inequalities
\begin{align}
	0 & \leq \frac{\mathrm{d}(\bar{f}(x))_i}{\mathrm{d}(x)_j} \leq (M_{f_{2}})_{i,j}-(M_{f_{1}})_{i,j} \text{ and } \label{aug_jac_f} \\ 
	0 & \leq \frac{\mathrm{d}(\bar{g}_{u}(x,u))_{i,k}}{\mathrm{d}(x)_j}\leq \left[(M_{g_{2}})_{i,j,k}-(M_{g_{1}})_{i,j,k}\right](u)_k,\label{aug_jac_g}
\end{align}
 where $i, j \coloneqq 1, \hdots, n$, and $k \coloneqq 1, \hdots, m$. Thus, $\bar{M_f}_1 = 0_{n \times n}$, $\bar{M_f}_2 = M_{f_{2}}-M_{f_{1}}$, $\bar{M_g}_1 = 0_{n \times n \times m}$ and $\bar{M_g}_2 = M_{g_{2}}-M_{g_{1}}$.
To simplify the notation, $n \times n \times 1$ arrays are treated as $n \times n$ matrices in the following development.

Using the derivative bounds, a state estimator with three correction terms is designed as
\begin{equation}
    	\dot{\hat{x}} = M_{f_{1}}\hat{x}  +M_{g_{1}}u\hat{x}+\bar{f}[\hat{x}+H\left(y-C\hat{x}\right)]+ \bar{g}_{u}[\hat{x}+K\left(y-C\hat{x}\right), u]+L\left(y-C\hat{x}\right),\label{eq:observerdynamics_x}
\end{equation}
where $\hat{x}  \in \mathbb{R}^{n}$ is the
estimate of $x$, $H \in \mathbb{R}^{n \times q}$, $K \in \mathbb{R}^{n \times q}$ and $L \in \mathbb{R}^{n \times q}$ are observer gains,  $H\left(y-C\hat{x}\right)$ and $K\left(y-C\hat{x}\right)$ are nonlinear injection terms and  $L\left(y-C\hat{x}\right)$ is a linear correction term. The estimation error is defined as
\begin{equation}
	e = x -\hat{x},\label{eq:dev_Estimation_error_x}
\end{equation}
and the estimation error dynamics is given by
\begin{equation}\label{aug_error}
	\dot{e} = \left( M_{f_{1}}+M_{g_{1}}u-LC\right)e+ \bar{f}(x)+\bar{g}_{u}(x,u)
    -\bar{f}[\hat{x}+H\left(y-C\hat{x}\right)]-\bar{g}_{u}[\hat{x}+K\left(y-C\hat{x}\right),u].
\end{equation}
For convenience of notation, let $ \bar{\phi}_{f}(t,e) \coloneqq   \bar{f}(x)-\bar{f}[\hat{x}+H\left(y-C\hat{x}\right)]$, $ \bar{\phi}_{g}(t,e, u)  \coloneqq  \bar{g}_{u}(x,u)-\bar{g}_{u}[\hat{x}+K\left(y-C\hat{x}\right),u]$, $ {M}_{ug} \coloneqq M_{g}u$, $\bar M_{ug_1} \coloneqq 0_{n \times n \times m}$, and $\bar M_{ug_2} \coloneqq \left(M_{g_{2}}-M_{g_{1}}\right)u$.

The differential mean value theorem (DMVT) in \cite[Theorem 2.1]{SCC.Zemouche.ea2005} guarantees that the difference function $\bar{\phi}_{f}(t,e)$ is proportional to $x-\hat{x}$, and can be expressed as
\begin{equation}
 \label{M_f}
\bar{\phi}_{f}(t,e) = \bar{M}_{f}(I-HC)(x-\hat{x}).
\end{equation}
where $\bar{M}_{f}$ is a time-varying matrix but always constrained in a compact set defined by $\bar{M}_{f_{1}}$, $\bar{M}_{f_{2}}$ in (\ref{aug_jac_f}).
Similarly, 
\begin{equation}
      \label{M_g}
        \bar{\phi}_{g}(t,e, u)  = \bar{M}_{g}u(I-KC)(x-\hat{x}),
\end{equation}
where the proportional factor $\bar M_{g}$ is a time-varying 3-dimensional array that is constrained in a compact set defined by $\bar{M}_{g_{1}}$, $\bar{M}_{g_{2}}$ in (\ref{aug_jac_g}).

\section{Multiplier Formulation and Sector Conditions}
\label{section:sectorFormulation}

In this section, the conditions sufficient for Lyapunov stability are derived by designing multiplier matrices that characterize the nonlinear functions $f$  and $g$. As shown in \cite{SCC.Behcet.Martin.ea2008}, using the multiplier matrix approach in the analysis and control of nonlinear systems, stability can be achieved if the  conditions developed in this section are satisfied.

The DMVT implies that the difference functions $ \bar{\phi}_{f}(t,e)$ and $ \bar{\phi}_{g}(t,e, u)$ are bounded as
\begin{equation}\label{eq:f_ineq}
 \bar{M}_{f_{1}}(I-HC)e\leq \bar{\phi}_{f}(t,e) \leq \bar{M}_{f_{2}}(I-HC)e,
\end{equation}
and
\begin{equation}\label{eq:g_ineq}
\bar{M}_{ug_{1}}(I-KC)e\leq \bar{\phi}_{g}(t,e, u) \leq \bar{M}_{ug_{2}}(I-KC)e.
\end{equation}
The stability of the state estimation error dynamics can now be shown using only the sector information about $ \bar{\phi}_{f}(t,e)$ and $ \bar{\phi}_{g}(t,e, u)$ constrained on a compact set $\mathcal{C}$, where the Jacobian bounds in  (\ref{aug_jac_f}) and (\ref{aug_jac_g}) hold. The bounds in (\ref{eq:f_ineq}) and (\ref{eq:g_ineq}) can be used to obtain the inequalities
\begin{equation}\label{quad_f}
    \left[ \bar{\phi}_{f}(t,e)\right]^\mathrm{T}[ \bar{\phi}_{f}(t,e)-\bar{M}_{f_{2}}(I-HC)e]\leq 0,
\end{equation}
and
\begin{equation}\label{quad_g}
  \left[ \bar{\phi}_{g}(t,e, u)\right]^\mathrm{T}[ \bar{\phi}_{g}(t,e, u)-\bar{M}_{ug_{2}}(I-KC)e]\leq 0.
\end{equation}
Rewriting the inequalities in (\ref{quad_f}) and (\ref{quad_g}) into their quadratic form yields,
\begin{equation}\label{matrixCondH}
\begin{bmatrix} 
e \\ \bar{\phi}_{f}
\end{bmatrix}^\mathrm{T}\begin{bmatrix}
I-HC & 0\\ 0 & I
\end{bmatrix}^\mathrm{T}J_{f}\begin{bmatrix}
I-HC & 0\\ 0 & I
\end{bmatrix}\begin{bmatrix}
e \\ \bar{\phi}_{f}
\end{bmatrix} \leq 0, 
\end{equation}
and
\begin{equation}\label{matrixCondK}
\begin{bmatrix}
e \\ \bar{\phi}_{g}
\end{bmatrix}^\mathrm{T}\begin{bmatrix}
I-KC & 0\\ 0 & I
\end{bmatrix}^\mathrm{T}J_{g}\begin{bmatrix}
I-KC & 0\\ 0 & I
\end{bmatrix}\begin{bmatrix}
e \\ \bar{\phi}_{g}
\end{bmatrix} \leq 0, 
\end{equation}
with 
\begin{equation}\label{J_f}
J_{f} = \begin{bmatrix}
0 & -\frac{M_{f_{2}}^\mathrm{T}-M_{f_{1}}^\mathrm{T}}{2}\\ -\frac{M_{f_{2}}-M_{f_{1}}}{2} & I
\end{bmatrix},
\end{equation}
and
\begin{equation}\label{J_g}
J_{g} = \begin{bmatrix}
0& -\frac{M_{ug_2}^\mathrm{T}-M_{ug_1}^\mathrm{T}}{2}\\ -\frac{M_{ug_2}-M_{ug_1}}{2} & I
\end{bmatrix}. 
\end{equation}
The observer error dynamics in (\ref{aug_error}) can now be expressed as
\begin{equation}\label{aug_error2}
     \dot{e} = \left(M_{f_{1}}+ M_{ug_1}-LC\right)e +\bar{\phi}_{f}(t,e) + \bar{\phi}_{g}(t,e,u).
\end{equation}
The following theorem establishes convergence of the estimator, provided the control input remains bounded and the system trajectories remain within the compact set $\mathcal{C}$ where the bounds in (\ref{aug_jac_f}) and (\ref{aug_jac_g}) hold.
\begin{remark}
    Note that the assumption that the trajectory and control signals are bounded applies only to the following theorem, and not to the controller designed in Section ~\ref{section:controlDesign}. The controller designed in Section ~\ref{section:controlDesign} ensures that provided the initial condition is close enough to the origin, the trajectories stay within the compact set $\mathcal{C}$. As such, provided the initial condition is close enough to the origin the bounds on the derivatives of $\bar f$ and $\bar g$ are guaranteed to hold along the trajectories of the closed-loop system.
\end{remark}
\begin{theorem}\label{thm:stateobserver}
    Given a system satisfying Assumption \ref{ass:jacobianbounds}, if there exists a symmetric positive definite matrix $P$, and observer gains $L$, $H$, $K$ that satisfy the matrix inequality
    \begin{equation}\label{lmi}
        \begin{bmatrix}
        (A-LC)^\mathrm{T} P + P (A-LC) + 2\alpha P & P - J_{21}^\mathrm{T} \\
        P - J_{21} & -(J_f)_{22} - (J_g)_{22}
        \end{bmatrix} \leq 0,
    \end{equation}
    where $A \coloneqq M_{f_{1}}+M_{ug_1}$, $J_{21} \coloneqq (J_{f})_{21}(I-HC)+(J_{g})_{21}(I-KC)$, and $I$ is an identity matrix, then the observer error system in (\ref{aug_error2}) is locally uniformly asymptotically stable.
    \end{theorem}
\begin{proof}
 Let $\mathcal{D}$ be an open subset of the set $ \{e \in\mathbb{R}^n: x,\hat{x} \in \mathcal{C}\}$ and consider the continuously differentiable candidate Lyapunov function, $V_{e} : \mathcal{D} \to  \mathbb{R}$ defined as
\begin{equation}
     V_e\left(e\right) = {e}^\mathrm{T}Pe, \label{eq:lyapunov function}
\end{equation}
which satisfies the following inequality
\begin{equation}\label{VeBound}
    \lambda_{\min}(P)\|e\|^2 \leq V_{e}\left(e\right) \leq \lambda_{\max}(P)\|e\|^2 
\end{equation}
where $\lambda_{\min}(\cdot)$ and  $\lambda_{\max}(\cdot)$ denote the minimum and maximum eigenvalues of a matrix, respectively. Since $P$ is a positive definite matrix, both eigenvalues are positive. 
On the set $\mathcal{D}$, the orbital derivative of the Lyapunov function along the trajectories of (\ref{aug_error}) can be expressed as
\begin{equation}
\dot{ V_e}\left(e\right) \ \coloneqq \ \begin{bmatrix}
e \\ \bar{\phi}_{f}
\end{bmatrix}^\mathrm{T}\begin{bmatrix}
\left(\begin{gathered}
\left(M_{f_{1}}-\frac{LC}{2}\right)^\mathrm{T}P\\
+P\left(M_{f_{1}}-\frac{LC}{2}\right)
\end{gathered}\right)
 &P\\P& 0
\end{bmatrix}  \begin{bmatrix}
e\\ \bar{\phi}_{f}
\end{bmatrix} +
 \begin{bmatrix}
e \\ \bar{\phi}_{g}
\end{bmatrix}^\mathrm{T}\begin{bmatrix}
\left(\begin{gathered}\left(M_{ug_1}-\frac{LC}{2}\right)^\mathrm{T}P\\
+P\left(M_{ug_1}-\frac{LC}{2}\right) \end{gathered}\right) & P\\ P& 0
\end{bmatrix}\begin{bmatrix}
e \\ \bar{\phi}_{g}
\end{bmatrix},
 \end{equation}
 Provided the matrix inequalities
% \cite{Boyd & Vandenberghe,2004}
\begin{equation}\label{non_lmi_f}
\begin{bmatrix}
\left(M_{f_{1}}-\frac{1}{2}\left(LC\right)\right)^\mathrm{T}P+P\left(M_{f_{1}}-\frac{1}{2}\left(LC\right)\right)+ \alpha P &P\\P& 0
\end{bmatrix} -
 \begin{bmatrix}
I-HC & 0\\ 0 & I
\end{bmatrix}^\mathrm{T}J_{f}\begin{bmatrix}
I-HC & 0\\ 0 & I
\end{bmatrix} \leq 0
\end{equation}
and 
\begin{equation}
\begin{bmatrix}\label{non_lmi_g}
\left(M_{ug_1}-\frac{1}{2}\left(LC\right)\right)^\mathrm{T}P+P\left(M_{ug_1}-\frac{1}{2}\left(LC\right)\right)+ \alpha P & P\\P& 0
\end{bmatrix} -
 \begin{bmatrix}
I-KC & 0\\ 0 & I
\end{bmatrix}^\mathrm{T}J_{g}\begin{bmatrix}
I-KC & 0\\ 0 & I
\end{bmatrix} \leq 0,
\end{equation}
are satisfied for some constant $\alpha > 0$, the multiplier matrices and sector conditions formulated in (\ref{matrixCondH}) and (\ref{matrixCondK}), and the S-Procedure Lemma \cite{SCC.Boyd1994} can be used to guarantee that the orbital derivative is bounded as (cf. \cite{SCC.Behcet.Martin.ea2008})
\begin{equation}\label{VeIneq}
    \dot{V}_{e}\left(e\right) \leq -2\alpha V_{e}\left(e\right), \forall e \in \mathcal{D}.
\end{equation}
Combining (\ref{non_lmi_f}) and (\ref{non_lmi_g}), the matrix inequality in (\ref{lmi}) is obtained.
%Hence, provided the LMI in (\ref{lmi}), needed to compute the symmetric positive definite matrix, $P$, and the three observer gains, $L$, $H$, and $K$, is feasible then the assumptions of \cite[Theorem~4.9]{SCC.Khalil2002} are satisfied with the inequality in (\ref{sLemmaCond}) bounded using (\ref{VeBound}) as
%\begin{equation}\label{VeIneq}
 %    \dot{V}\left(e\right) \leq -2\alpha V\left(e\right) \leq -\lambda_{\min}(P)\|e\|^2,\forall e \in \mathcal{D}.
%\end{equation}
Using the bound in \eqref{VeIneq}, it can be concluded that the origin of the error system, $e = 0$, is locally uniformly asymptotically stable.
In particular, let $r \in \mathbb{R}_{> 0}$ be a constant such that, $B_{r} \coloneqq \{e\in\mathbb{R}^n \mid \|e\| \leq r\} \subset \mathcal{D}$
and, $W_{2}(\|e\|)\coloneqq \lambda_{\min}(P)\|e\|^2$. Select $c > 0$ such that $c < \frac{r^2\lambda_{\min}(P)}{2}$, Theorem 4.9 in \cite{SCC.Khalil2002} can then be invoked to conclude that every trajectory starting in $\{e \in B_{r} \mid W_{2}(\|e\|) \leq c\}$ stays within $\mathcal{D}$ for all $t\geq 0$ and satisfies
\begin{equation}
    \|e(t)\| \leq \beta(\|e(t_{0})\|,t-t_{0}), \forall t \geq t_{0} \geq 0
\end{equation}
where $\beta$ is a class $\mathcal{KL}$ function.
\end{proof}
\begin{remark}
    The matrix inequality can be reformulated as a linear matrix inequality (LMI) using the typical variable substitution method. Indeed, substituting $L = P^{-1}R$ in \eqref{lmi}, the matrix $P$ and the observer gains $L$, $H$, and $K$ can be obtained by solving the LMI
    \begin{equation}\label{lmi2}
    \begin{bmatrix}
    A^\mathrm{T}P+PA-C^\mathrm{T}R^\mathrm{T}-RC + 2\alpha P & P-{J_{21}}^\mathrm{T} \\P-J_{21}  & -\left((J_f)_{22}+ (J_g)_{22}\right)
    \end{bmatrix} \leq 0
    \end{equation}
    for $P$, $R$, $H$, and $K$.
\end{remark}
\begin{remark}
    The observer design is only valid if the control input remains bounded and the system trajectories remain within the compact set $\mathcal{C}$ where the bounds on the derivatives of $\bar f$ and $\bar g$, in (\ref{aug_jac_f}) and (\ref{aug_jac_g}), respectively, are valid.  In the theorem above, the derivative bounds are local, and as a result the observer error is locally uniformly asymptotically stable.  If the derivative bounds hold globally, then a similar analysis can be used to show that the observer error is globally uniformly asymptotically stable. The controller designed in Section ~\ref{section:controlDesign} ensures that provided the initial condition is close enough to the origin, the trajectories stay within the desired compact set $\mathcal{C}$.
\end{remark}

\section{Control Design}
\label{section:controlDesign}
To achieve the control objective stated above  while satisfying all constraints of the system, the cost functional to be minimized is given as
\begin{equation} \label{costfunction}
J(x,u(\cdot)) \coloneqq 	\int_{0}^\infty  Q(\phi(\tau, x, u_{[t,\infty)}(\cdot))) + U(u(\tau))d\tau,
\end{equation} over the set $\mathcal{U} $ piecewise continuous functions $t \to u(t)$, $\forall t \in [0, \infty)$ 
where $\phi(t, x, u(\cdot))$ is a solution of (\ref{eq:dynamics_x}) under control signal $u(\cdot)$ starting from $x$, $Q:\mathbb{R}^{n}\to\mathbb{R}$ is a continuous, positive definite function and $U:\mathbb{R}^{m}\to\mathbb{R}$, introduced to address the saturation constraint on the control, is defined as
\begin{equation} \label{eq: positive function}
    U(u) \coloneqq 2	\int_{0}^{u} (\bar{\lambda}\tanh^{-1}(\upsilon/\bar{\lambda}))^\mathrm{T}Rd\upsilon,
\end{equation}
where $R \coloneqq \diag(r_1,\hdots, r_m).$
Assuming the optimal controller exists, then let the optimal value function, $V^{*}: \mathbb{R}^{n} \times \mathbb{R}^m  \rightarrow \mathbb{R} $, be expressed as  \begin{equation}V^{*}(x):= \min_{u(\cdot)\in \mathcal{U}_{[t,\infty)}}\int_{t}^\infty Q(\phi(\tau, x, u_{[t,\tau)}(\cdot))) + U(u(\tau))d\tau, \label{eq:valuefunction}\end{equation} where $u_I$ and $\mathcal{U}_I$ are obtained by restricting the domains of $u$ and functions in $\mathcal{U}_I$ to the interval
$I \subseteq R$, respectively.
 Assuming that the optimal value function is continuously differentiable, it can be shown to be the unique PD solution of the Hamilton-Jacobi-Bellman (HJB) equation, \cite[Theorem 1.5]{SCC.Kamalapurkar.Walters.ea2018},
\begin{equation}\label{HJB} 
\min_{u\in\mathbb{R}^m} \Big(\nabla_x V\left(f(x)+g(x)u\right)+Q(x)+U(u)\Big) = 0,
\end{equation}
where $\nabla_{\left(\cdot\right)} \coloneqq  \frac{\partial}{\partial \left(\cdot\right)}$.\\
Therefore, the optimal controller is given by the feedback policy, $u(t) = u^*(\phi(t, x, u_{[0,t)}))$  where $u^{*}: \mathbb{R}^n \to \mathbb{R}^m$ defined as
\begin{equation}\label{eq:optimalcontrol}
    u^{*}(x) := -\bar{\lambda}\tanh(D^*),
\end{equation}
where $D^*=(1/2\bar{\lambda})R^{-1}g(x)^\mathrm{T}\nabla_{x}V^{*}(x)\in\mathbb{R}^m$. Substituting  equation (\ref{eq:optimalcontrol}) in (\ref{eq: positive function}), the function $U$ is given as
\begin{equation} \label{eq:optimalstarU}
U(u^*) \coloneqq \bar{\lambda}\nabla_{x}{V^{*}}^\mathrm{T}(x)g(x)\tanh(D^*)+\bar{\lambda}^2\bar{R}\ln(\boldsymbol{1}-\tanh^2(D^*)),
\end{equation}
where $\bar{R} \coloneqq [r_1,\hdots,r_m] \in \mathbb{R}^{1\times m}$ and $\boldsymbol{1}$ denotes a column vector having all of its elements equal to one.
Substituting optimal control input, (\ref{eq:optimalcontrol}) in  (\ref{HJB}), the following equation is obtained,
\begin{equation}\label{eq:optimalPerformFunc}
   \nabla_x V^*\left(f(x)+g(x)u^{*}(x)\right)+Q(x)+U(u^*) = 0
\end{equation}

\subsection{Value Function Approximation }

Solving the above HJB equation is generally infeasible due to its inherent non-linearity, hence to find an approximate solution, estimates of the value function and the control policy are introduced. Using  the value function can be expressed as
\begin{equation} \label{eq:optimalV}
    V^{*}\left(x\right)=W^\mathrm{T}\sigma\left(x\right)+\epsilon\left(x\right),
\end{equation}
where
\begin{equation} \label{eq:gradOptimalV}
    \nabla_{x}V^{*}\left(x\right)= {\nabla_{x}\sigma}^\mathrm{T}\left(x\right)W+\nabla_{x}\epsilon\left(x\right),
\end{equation}
$W\in\mathbb{R}^{L}$ is an unknown vector of bounded weights, $\sigma:\mathbb{R}^{2n}\rightarrow\mathbb{R}^{L}$ is a vector of continuously differentiable nonlinear activation functions such that $\sigma\left(0\right)=0$ and $\nabla_{x} \sigma \left(0\right)=0$, $L\in\mathbb{N}$ is the number of basis functions, and $\epsilon:\mathbb{R}^{2n}\rightarrow\mathbb{R}$ is the reconstruction error. Using the Stone-Weierstrass Theorem\cite[Theorem 1.5]{SCC.Sauvigny2012}, the activation functions $\sigma$ can be selected so that the weights and the approximation errors satisfy $\sup_{x \in \mathcal{C}}\|W\| \leq \bar{W}$, $ \sup_{x \in \mathcal{C}}\|\sigma(\cdot)\| \leq \bar{\sigma}$, $ \sup_{x \in \mathcal{C}}\|\nabla_{(\cdot)}\sigma(\cdot)\| \leq \bar{\sigma}, \sup_{x \in \mathcal{C}}\|\epsilon(\cdot)\| \leq \bar{\epsilon}$ and $\sup_{x \in \mathcal{C}}\|\nabla_{(\cdot)}\epsilon(\cdot)\| \leq \bar{\epsilon}$

The optimal controller can be then expressed as
\begin{equation} \label{eq:optimalU}
u^*\left(x\right)=-\bar{\lambda}\tanh\left((1/2\lambda)R^{-1}g(x)^\mathrm{T}\nabla_{x}V^{*}\left(x\right)\right).
\end{equation}
Since the ideal weights, $W$, are unknown, let the estimates $\hat{V}:\mathbb{R}^{n}\times\mathbb{R}^{L}\to\mathbb{R}$ and $\hat{u}:\mathbb{R}^{n}\times\mathbb{R}^{L}\to\mathbb{R}^{m}$ be defined as
\begin{gather}
\hat{V}\left(\hat{x},\hat{W}_{c}\right)\coloneqq\hat{W}_{c}^\mathrm{T}\sigma\left(\hat{x}\right),\label{V_app}\\
\hat{u}\left(\hat{x},\hat{W}_{c}\right)\coloneqq-\bar{\lambda}\tanh(\hat{D}),\label{u_app}
\end{gather}
where $\hat{D}= \frac{1}{2\lambda}R^{-1}g(\hat{x})^\mathrm{T}\nabla_{\hat{x}}\sigma(\hat{x})^\mathrm{T}\hat{W}_{c}.$ The critic weights, $\hat{W}_{c}\in\mathbb{R}^{L}$ are an estimate of the ideal weights, $W$. 

Substituting (\ref{V_app}) and (\ref{u_app}) into (\ref{HJB}) the residual term, $\hat{\delta}: \mathbb{R}^{2n} \times \mathbb{R}^{L} \times \mathbb{R}^{L} \rightarrow \mathbb{R}$, referred to as the Bellman error (BE), is obtained as
\begin{equation} \label{BE1}
    \hat{\delta}(\hat{x},\hat{W}_{c}) \coloneqq  \nabla_{x}\hat{V}(\hat x,\hat{W}_{c})\left(f(\hat{x})+g(\hat{x})\hat{u}(\hat{x},\hat{W}_{c})\right) + U(\hat{u}) + Q(\hat{x}).
\end{equation}
Substituting the approximate form of the HJB in (\ref{eq:optimalPerformFunc}) into (\ref{BE1}), $
    \hat{\delta}(\hat{x},\hat{W}_{c}) =  \nabla_{x}\hat{V}(\hat{x} ,\hat{W}_{c})\left(f(\hat{x})+g(\hat{x})\hat{u}(\hat{x},\hat{W}_{c})\right) + U(\hat{u})- \nabla_x V^*\left(f(\hat{x})+g(\hat{x})u^{*}\right)-U(u^*).$
    
For notational brevity, the dependence of functions on $x$ and $u$ is omitted hereafter.  Substituting  (\ref{V_app}) in the equation above, and using fact that $\ln(\boldsymbol{1}-\tanh^2(D^{*})) = \ln(4)-2D^{*}\sign(D^{*})+\varepsilon_D^{*}$ and $\ln(\boldsymbol{1}-\tanh^2(\hat{D})) = \ln(4)-2\hat{D}\sign(\hat{D})+\varepsilon_{\hat{D}}$, where $\varepsilon_D$ is the error of approximating a $\sign$ function with a $\tanh$ function and  satisfies$\|\hat\varepsilon_D\| \leq \ln(4)$. Hence, $\ln(\boldsymbol{1}-\tanh^2(\hat{D}))-\ln(\boldsymbol{1}-\tanh^2(D^{*})) = 2D^*\sign(D^*)-2\hat{D}\sign(\hat{D})+{\varepsilon}_{\hat{D}}-{\varepsilon}_{D^*}$.
The error $\hat{U}-U^*$ can thus be expressed as 
\begin{multline}
    \hat{U}-U^* = 2\bar{\lambda}^2\bar{R}\left(D^*\sign(D^*) -\hat{D}\sign(\hat{D})\right) + \bar{\lambda}\hat{W}_c^\mathrm{T}\nabla\sigma \hat{g}\tanh(\hat{D})-\bar{\lambda}W^\mathrm{T}\nabla\sigma \hat{g}\tanh(D^*)- \bar{\lambda}\nabla\epsilon\hat{g}\tanh(D^*)\\+\bar{\lambda}^2\bar{R}({\varepsilon}_{\hat{D}} -{\varepsilon}_{D^*}).\end{multline}
Since $D^*\sign(D^*) = |D^*|$ and $-D^* \leq D^*\tanh(D^*) \leq D^*$, then  $D^*\sign(D^*)-D^*\tanh(D^*)= C_{D^*}$, where the error $C_{D^*}$ satisfies $C_{D^*} \geq 0$. Using the same fact,  $\hat{D}\sign(\hat{D})= C_{\hat{D}}$ which satisfies $C_{\hat{D}} \geq 0$. Thus, the difference can be expressed as, $D^*\sign(D^*)-\hat{D}\sign(\hat{D}) = \frac{1}{2}R^{-1}W\nabla\sigma\hat{g}\tanh(D^*)+\frac{1}{2}R^{-1}\nabla\epsilon\hat{g}\tanh(D^*)-\frac{1}{2}R^{-1}\hat{W}_c\nabla\sigma\hat{g}\tanh(\hat{D}) + C_{D^*}-C_{\hat{D}}$. Hence, the BE can be expressed as
\begin{equation}
    \hat{\delta} \coloneqq -\omega^\mathrm{T}\tilde{W}_c+ \Delta,
\end{equation}
where $\omega \coloneqq \nabla\sigma \left(\hat{f}+\hat{g}\hat{u}\right)$ and $\Delta \coloneqq -\nabla\epsilon{\left(\hat{f}+\hat{g}u^*\right)}+\bar{\lambda}W^\mathrm{T}\nabla\sigma\hat{g}\left(\tanh(D^*)- \tanh(\hat{D})\right)+2\bar{\lambda}^2\bar{R}(C_{D^*}-C_{\hat{D}})+\bar{\lambda}^2\bar{R}({\varepsilon}_{\hat{D}}-{\varepsilon}_{D^*}).$

To accurately approximate the value function, online RL methods require persistence of
excitation (PE) condition \cite{SCC.Modares.Lewis.ea2013, SCC.Kamalapurkar.Rosenfeld.ea2016}. which is difficult to guarantee in practice since the
 system needs to explore several points in the state space. However, through BE extrapolation for excitation via simulation, stability, and convergence of online RL can be established using Assumption~\ref{ass:CLBCADPLearnCond}.
To simulate experience using BE extrapolation,  select a set of trajectories $\left\{ x_{i}: \mathbb{R}_{\geq t_0} \to \mathbb{R}^n \mid i=1,\cdots, N\right\}$ and extrapolate the BE along these trajectories to yield the BEs, $\hat{\delta}: \mathbb{R}^{2n} \times \mathbb{R}^{L} \times \mathbb{R}^{L} \rightarrow \mathbb{R}$, given by
 \begin{equation} \label{BE1_i}
    \hat{\delta_i}(x_i,\hat{W}_{c}) \coloneqq  \nabla_{x}\hat{V}(x_i,\hat{W}_{c})\left(f\left(x_i\right)+g\left(x_i\right)\hat{u}(x_i,\hat{W}_{c})\right)+U(\hat{u}) + Q(x_i).
\end{equation}
Given the critic weight estimation error $\tilde{W}_{c} \coloneqq W -\hat{W}_{c}$ and substituting \eqref{V_app} and \eqref{u_app} into (\ref{HJB}), and subtracting from \eqref{BE1}, the BE can be expressed as
\begin{equation}
    \hat{\delta_i} \coloneqq -\omega_i^\mathrm{T}\tilde{W}_c+ \Delta_i,
\end{equation}
where $\hat{f}_i \coloneqq  f\left(x_i\right)$, $\hat{g}_i \coloneqq g\left(x_i\right)$, $\sigma_{i} \coloneqq \sigma (x_{i})$, $\omega_i \coloneqq \nabla\sigma_i \left(\hat{f}_i+\hat{g}_i\hat{u}(x_i,\hat{W}_{c})\right)$, $\Delta_i \coloneqq -\nabla\epsilon_{i}{\left(\hat{f}_i+\hat{g}_iu^{*}(x_{i})\right)} \\+\bar{\lambda}W^\mathrm{T}\nabla\sigma_{i}\hat{g}_i  \left(\tanh({D_i}^*)- \tanh(\hat{{D}_{i})}\right)+2\bar{\lambda}^2\bar{R}(C_{{D_i}^*}-C_{\hat{D}_{i}})+\bar{\lambda}^2\bar{R}({\varepsilon}_{\hat{{D_i}}}-{\varepsilon}_{{D_i}^*})$, $\nabla\epsilon_{i} = \nabla\epsilon(x_{i})$.

\subsection{Update laws for Critic weights}
To guarantee that the estimated value function weights, $\hat{W}_c$, converge to their ideal weights in (\ref{eq:optimalV}), the  estimated value function weights are updated based on the result of the stability analysis in Section~\ref{section:stabilityAnalysis} as
\begin{align}
    \dot{\hat{W}}_{c} &=- \frac{k_{c}}{N}\Gamma\sum_{i=1}^{N}\frac{\omega_{i}}{\rho_{i}}\hat\delta_{i},\label{W_c}\\
    \dot{\Gamma} &= \beta\Gamma- \frac{k_{c}}{N}\Gamma\sum_{i=1}^{N}\frac{\omega_{i}\omega_{i}^\mathrm{T}}{\rho_{i}^{2}}\Gamma,\label{gamma}
\end{align}
with $\Gamma\left(t_{0}\right)=\Gamma_{0}$, where $\Gamma:\mathbb{R}_{\geq t_{0}} \to \mathbb{R}^{L\times L}$
is a time-varying least-squares gain matrix, $\rho_{i}\left(t\right)\coloneqq 1+\gamma\omega_{i}^\mathrm{T}\left(t\right)\omega_{i}\left(t\right)$, $\gamma > 0$ is a constant positive normalization gain,  $\beta > 0 \in \mathbb{R}$ is a constant forgetting factor, and $k_{c} > 0 \in \mathbb{R}$ is a constant adaptation gain.

\section{Stability Analysis}\label{section:stabilityAnalysis}
In this section, stability analysis of the observer-controller RL architecture will be carried out using Lyapunov methods.
To facilitate the stability analysis, the following verifiable  persistence of
excitation (PE) rank condition is utilized in the
stability analysis
\begin{assum}
    \label{ass:CLBCADPLearnCond}There exists a constant $\underline{c}_{1}$ such that the finite set of trajectories $\left\{x_{i}: \mathbb{R}_{\geq t_0} \mid i=1,\hdots,N\right\}$satisfies
    \begin{equation}
    0 < \underline{c}_{1} \leq\inf_{t\in\mathbb{R}_{\geq T}}\lambda_{\min}\left(\frac{1}{N}\sum_{i=1}^{N}\frac{\omega_{i}\left(t\right)\omega_{i}^\mathrm{T}\left(t\right)}{\rho_{i}^{2}\left(t\right)}\right).\label{eq:CLBCPE2}
    \end{equation}
\end{assum}

As described in\cite{SCC.Mahmud.Nivison.ea2021}, since $\omega_{i}$ is a function of the $x$ and $\hat{W}_{c}$,  Assumption \ref{ass:CLBCADPLearnCond} cannot be guaranteed a priori. However, unlike the PE condition utilized in \cite{SCC.Vamvoudakis.Lewis2010}, Assumption \ref{ass:CLBCADPLearnCond} can be verified online. Furthermore, since $\lambda_{\min}\left(\sum_{i=1}^{N}\frac{\omega_{i}\left(t\right)\omega_{i}^\mathrm{T}\left(t\right)}{\rho_{i}^{2}\left(t\right)}\right)$ is non-decreasing in the number of samples, $N$, Assumption \ref{ass:CLBCADPLearnCond} can be met, heuristically, by increasing the number of samples. The calculation of a precise bound on the number of samples is out of the scope of this paper.

 Let $Z \coloneqq [x^\mathrm{T}, e^\mathrm{T}, {\tilde{W}_c}]^\mathrm{T}$ represent the concatenated state of the closed-loop system and let a continuously differentiable candidate Lyapunov function, $V_L: R^{2n+L} \times \mathbb{R}_{\geq 0} \to \mathbb{R}$, be defined as,
 \begin{equation}\label{eq:Lyap}
     V_L\left(Z, t\right) \coloneqq V^*\left(x\right) + \frac{1}{2}\tilde{W}_{c}^\mathrm{T}\Gamma^{-1}(t)\tilde{W}_{c} + V_{e}\left(e\right),
 \end{equation}
 where $V^*$ represent the optimal value function and $V_e$ is introduced in (\ref{eq:lyapunov function}) in Section ~\ref{section:stateEstimator}. 
 To facilitate the stability analysis, let $\underline{c} \in \mathbb{R}_{>0}$ be a constant defined as
\begin{equation} \label{c_}
       \underline{c} \coloneqq\frac{\beta}{2\Gamma k_c} + \frac{\underline{c}_1}{2},
\end{equation}
and $\iota \in \mathbb{R}$ be a positive constant defined as 
\begin{equation}
      \iota = \frac{L_{g\sigma}^2\overline{W}^2}{2\bar{\lambda}^2\lambda_{\min}(P)}+ \frac{3\overline{\|G_{r\sigma}\|}^2}{4\lambda^2k_c\underline{c}}+(1/2\lambda)\overline{\|G_{r}\|}\|\overline{\|\nabla\epsilon\|}+\bar{\lambda}L_g\overline{\|\nabla\epsilon\|}+\frac{3k_{c}}{4\underline{c}}\|\frac{\omega_{i}}{\rho_{i}}\|^2\|\Delta_{i}\|^2.
 \end{equation}
As shown in \cite[Lemma~1]{SCC.Kamalapurkar.Rosenfeld.ea2016}, provided (\ref{ass:CLBCADPLearnCond}) holds and $\lambda_{\min}\{{\Gamma_{0}^{-1}}\}> 0$, the update law in (\ref{gamma}) ensures that the least squares update law satisfies
\begin{equation}\label{eq:OFBADP1Gammabound}
	\underline{\Gamma}I_{L}\leq\Gamma\left(t\right)\leq\overline{\Gamma}I_{L},		\end{equation}
$\forall t\in\mathbb{R}_{\geq 0}$ and for some  $\overline{\Gamma},\underline{\Gamma}>0$. Using the bound in \eqref{eq:OFBADP1Gammabound} and since the candidate Lyapunov function is  positive definite, \cite[Lemma 4.3]{SCC.Khalil2002} can be used to conclude that it is bounded as
\begin{equation}
\underline{v}\left(\left\Vert Z\right\Vert \right)\leq V_{L}\left(Z,t\right)\leq\overline{v}\left(\left\Vert Z\right\Vert \right),\label{eq:OFBADPVBound}
\end{equation}
for all $t \in \mathbb{R}_{\geq 0}$ and for all $Z\in\mathbb{R}^{2n+L}$, where $\underline{v},\overline{v}:\mathbb{R}_{\geq 0}\rightarrow\mathbb{R}_{\geq 0}$ are class $\mathcal{K}$ functions. Let $\upsilon_{l}: \mathbb{R}_{\geq 0} \to \mathbb{R}_{\geq 0}$ be a class $\mathcal{K}$ function such that \begin{equation}
    \upsilon_{l}\left(\|Z\|\right) \leq \frac{\lambda_{\min}(Q)}{2}\|x\|^2+\frac{k_{c}\underline{c}}{6}\|\tilde{W}_{c}\|^2 + \frac{\lambda_{\min}(P)}{4}\|e\|^2. 
\end{equation}
\begin{theorem}
   Provided Assumptions \ref{ass:jacobianbounds} and \ref{ass:CLBCADPLearnCond} hold, the conditions specified in Theorem~\ref{thm:stateobserver} are satisfied, the control gains are selected large enough based on the sufficient condition \footnote{Despite the fact that $\iota$ generally increases with increasing $\zeta$, the condition in (\ref{zeta_cond}) can be satisfied provided the points for BE extrapolation are selected such that $\underline{c}$, introduced in (\ref{c_}) and control gain, $k_{c}$ is large enough, and the basis for the value function approximation are selected such that $\overline{\|\epsilon\|}$ and $\overline{\|\nabla{\epsilon\|}}$ are sufficiently small.} 
\begin{gather}
  {\upsilon_l}^{-1}\left(\iota\right) \leq  {\overline{\upsilon}}^{-1}\left(  \underline{\upsilon}\left(\zeta\right)\right),\label{zeta_cond}
\end{gather}
and the weights $\hat{W}_c$ and $\Gamma$ are updated according to (\ref{W_c}) and (\ref{gamma}), respectively, then the concatenated state, $Z$, is locally uniformly ultimately bounded under the controller designed in (\ref{u_app}).
\end{theorem}
\begin{proof}
% The bound of $f$ and the NN function approximation error depend on the underlying compact set, hence, $\iota$ is a function of $\zeta$. 
Let $\chi \subset \mathcal{C} \times \mathcal{D} \times \mathbb{R}^L$ be an open set. The orbital derivative of the candidate Lyapunov function, $V_L$, along the trajectories of (\ref{eq:dynamics_x}), (\ref{aug_error}), (\ref{W_c}) and (\ref{gamma})  is given by,
\begin{equation}\label{eq:LyapD}
     \dot{V}_L\left(Z, t\right) =  {\nabla}{V^*}\dot{x} - \tilde{W}_{c}^\mathrm{T}\Gamma^{-1}\dot{\hat{W}}_{c}-\frac{1}{2} \tilde{W}_{c}^\mathrm{T}\Gamma^{-1}\dot{\Gamma}\Gamma^{-1}\tilde{W}_{c} + \dot{V}_e.
\end{equation}
Substituting (\ref{eq:dynamics_x}), (\ref{W_c}) and (\ref{gamma}) in (\ref{eq:LyapD}), using (\ref{M_f}), (\ref{eq:optimalPerformFunc}), (\ref{M_g}), (\ref{eq:optimalcontrol}), (\ref{u_app}), (\ref{VeIneq}), and the fact that $\frac{\omega_{i}{\omega_{i}}^\mathrm{T}}{\rho_{i}^2}\leq \frac{\omega_{i}{\omega_{i}}^\mathrm{T}}{\rho_{i}}$, the orbital derivative can be bounded, on the set $\chi\times\mathbb{R}_{\geq 0}$, as
 %  \begin{multline}
 %     \dot{V}_L = {\nabla}{V^*}\left(f+g\hat{u}\right) \\+ \tilde{W}_{c}^\mathrm{T}\Gamma^{-1}\left(\frac{k_{c}}{N}\Gamma\sum_{i=1}^{N}\frac{\omega_{i}}{\rho_{i}}\left(-\omega_i^\mathrm{T}\tilde{W}_c+ \Delta_i\right)\right)\\-\frac{1}{2} \tilde{W}_{c}^\mathrm{T}\Gamma^{-1}\left(\beta\Gamma-\frac{k_{c}}{N}\Gamma\sum_{i=1}^{N}\frac{\omega_{i}\omega_{i}^\mathrm{T}}{\rho_{i}^{2}}\Gamma\right)\Gamma^{-1}\tilde{W}_c +\dot{V}_e
 % \end{multline}
 % It is known that $\frac{\omega_{i}{\omega_{i}}^\mathrm{T}}{\rho_{i}^2}\leq \frac{\omega_{i}{\omega_{i}}^\mathrm{T}}{\rho_{i}}$. Using (\ref{VeIneq}),   (\ref{M_f}),  (\ref{M_g}), (\ref{eq:optimalPerformFunc}), (\ref{eq:optimalcontrol}), (\ref{u_app}) and (\ref{VeIneq}), $\dot{V}_L$ can be expressed as,
\begin{multline} \label{eq:dot_V_l}
     \dot{V}_L(Z,t) \leq -Q(x)-U(u^*)+\bar{\lambda}W^\mathrm{T}{\nabla}\sigma g(x)\left(\tanh(D^*)-\tanh(\hat{D})\right) +\bar{\lambda}\nabla\epsilon g(x) \tanh(D^*)- \frac{k_{c}}{2N}\tilde{W}_{c}^\mathrm{T}\sum_{i=1}^{N}\frac{\omega_{i}{\omega_{i}}^\mathrm{T}}{\rho_{i}}\tilde{W}_c \\+
     \frac{k_{c}}{N}\tilde{W}_{c}^\mathrm{T}\sum_{i=1}^{N}\frac{\omega_{i}}{\rho_{i}}\Delta_{i}-\frac{1}{2}\beta \tilde{W}_{c}^\mathrm{T}\Gamma^{-1}\tilde{W}_c - \lambda_{\min}(P)\|e\|^2.
 \end{multline}
 The bound in \eqref{eq:dot_V_l} utilizes the facts that $U(u^*)$ is a positive definite function and
\begin{equation}
     \left\|\tanh(D^*) - \tanh(\hat{D})\right\| \leq  (1/2\lambda)L_{g\sigma}\overline{W}\|e\|+(1/2\lambda)\overline{\|G_{r\sigma}\|}\|\tilde{W}_{c}\| +(1/2\lambda)\overline{\|G_{r}\|}\|\overline{\|\nabla\epsilon\|},
\end{equation}
 where $G_{r\sigma}\left(x\right)\coloneqq R^{-1}g(x)^\mathrm{T}{\nabla_{x}\sigma}^\mathrm{T}\left(x\right)$, $G_{r\sigma}\left(\hat{x}\right) \coloneqq R^{-1}g(\hat{x})^\mathrm{T}\nabla_{\hat{x}}\sigma(\hat{x})^\mathrm{T}$, $G_{r}\left(x\right) \coloneqq R^{-1}g(x)^\mathrm{T}$, and $L_{g\sigma}$ denotes the Lipschitz constant of $G_{r\sigma}$ over the set $\chi$. 
  
Completing the squares and using the Cauchy Schwartz inequality, the orbital derivative is bounded, on the set $\chi\times\mathbb{R}_{\geq 0}$, as
\begin{equation}
     \dot{V}_L(Z,t) \leq -\lambda_{\min}(Q)\|x\|^2-\frac{k_{c}\underline{c}}{3}\|\tilde{W}_{c}\|^2 - \frac{\lambda_{\min}(P)}{2}\|e\|^2 + \iota.
 \end{equation}
Let $\zeta$ be a constant such that $B_{\zeta}\subset\chi$. Based on the conditions stated in (\ref{zeta_cond}) and (\ref{VeIneq}), the orbital derivative can be bounded as 
\begin{equation}
     \dot{V}_L(Z,t) \leq -\upsilon_{l}\left(\|Z\|\right),  \forall {\upsilon_{l}}^{-1}(\iota) < \|Z\| < \zeta , \forall t\geq 0.
 \end{equation}
Using the sufficient condition  stated in  (\ref{zeta_cond}), \cite[Theorem~4.18]{SCC.Khalil2002} can be invoked to conclude that $Z$ is locally uniformly ultimately bounded. In particular, all trajectories starting from initial conditions bounded by $\|Z(0)\| \leq  {\overline{\upsilon}}^{-1}\left(  \underline{{\upsilon_l}}\left(\zeta\right)\right)$ remain with $\chi$ for all $t\geq 0$ and satisfy $\lim \sup_{t\to\infty} \|Z(t)\| \leq  {\overline{\upsilon}}^{-1}\left(  \underline{{\upsilon_l}}\left(\iota\right)\right)$. Therefore, provided $\|Z(0)\| \leq  {\overline{\upsilon}}^{-1}\left(  \underline{{\upsilon_l}}\left(\zeta\right)\right)$, the state and the state estimates, under the controller in \eqref{u_app} and the observer in \eqref{eq:observerdynamics_x}, remain within the compact set $\mathcal{C}$ where the Jacobian bounds and the Lipschitz constants are valid.
\end{proof}

\section{Simulation Results}\label{section:simulation}
\subsection{Two state dynamical system}\label{simsec1_SE}
In this section, a simulation study is performed to demonstrate the effectiveness of the developed method on a control affine nonlinear system of the form (\ref{eq:dynamics_x}) with states $x$ = $[(x)_1;(x)_2]$, where
 \begin{align}\label{sim_dyn}
f(x) &=\begin{bmatrix}
-(x)_1 + (x)_2 \\
 -(x)_1-\frac{1}{2}(x)_2\left(1-\left(\cos\left(2(x)_1\right)+2\right)^2\right)
\end{bmatrix}, \nonumber\\
g(x) &=\begin{bmatrix}
0 \\
\cos\left(2(x)_1\right)+2
\end{bmatrix}, C = \begin{bmatrix}
    0 & 1 
\end{bmatrix}
% \end{align}
\end{align}
In this study, the system dynamics are assumed to be unknown and the system states are not available for measurement, the output of the system is, however, measurable. To obtain the symmetric positive definite matrix, P, and the three observer gains, $L$, $H$, and $K$ that satisfy the stability conditions developed in Section ~\ref{section:sectorFormulation}, the LMI in (\ref{lmi}) is solved using Sedumi in YALMIP on Matlab  with the learning rate, $\alpha = 2$.

The control objective is to minimize the infinite horizon cost in (\ref{costfunction}), with $Q(x) = x^\mathrm{T}Qx$, where $Q = 5I_{2}$ and $R = 1$.  The basis for value function approximation is selected as $\sigma(\hat{x}) = 
[(\hat{x})_{1}^2;(\hat{x})_{1}(\hat{x})_{2} ;(\hat{x})_{2}^2]$. The initial conditions of the states, the estimated states, the critic weights, and the least squares gain matrix are selected as $x(0) = [-1;1]$, $\hat{x}(0) = [2;1.5]$, $\hat{W_{c}}(0) = [0.4;0.2;0.8]$, $\Gamma (0) = 50I_{3}$. The saturation constraint on the control input is selected as, $\bar{\lambda} = 3$, and the learning gains are selected as $k_c = 0.01$, $\nu=0.7$, and $\beta = 0.2$. The simulation uses 100 fixed Bellman error extrapolation points selected from a $2\times 2$  square centered around the origin of the system. 
% \vspace{-5mm}
\begin{figure}[H]
        \centering
		\includegraphics[width=1\columnwidth]{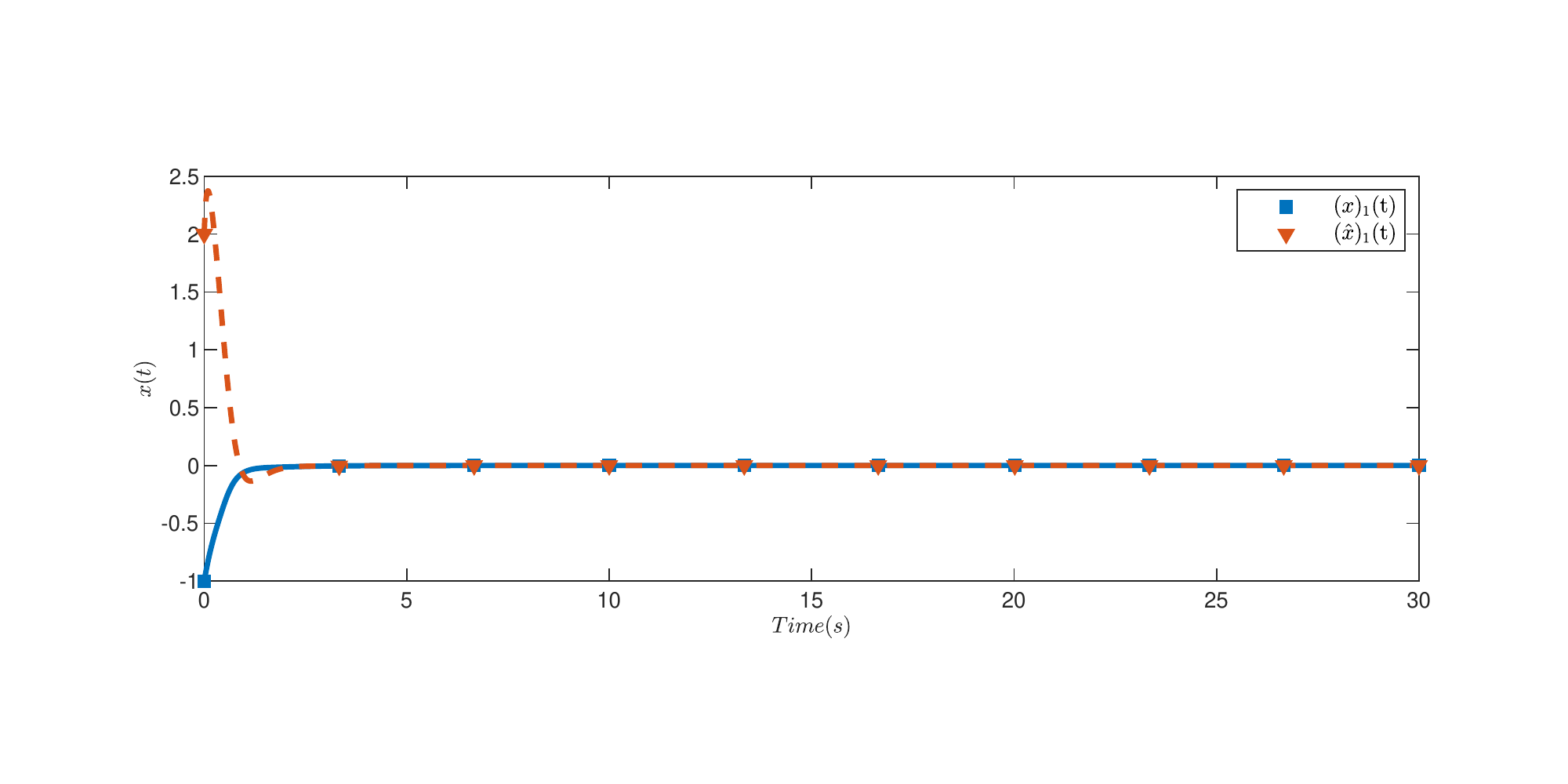}
            \vspace{-15mm}
		\caption{The trajectories of the actual state $x_1$ and estimated state $\hat{x}_1$.}
		\label{fig:state1_sim1}
\end{figure}
\begin{figure}[H]
  \vspace{-3mm}
\centering	\includegraphics[width=1\columnwidth]{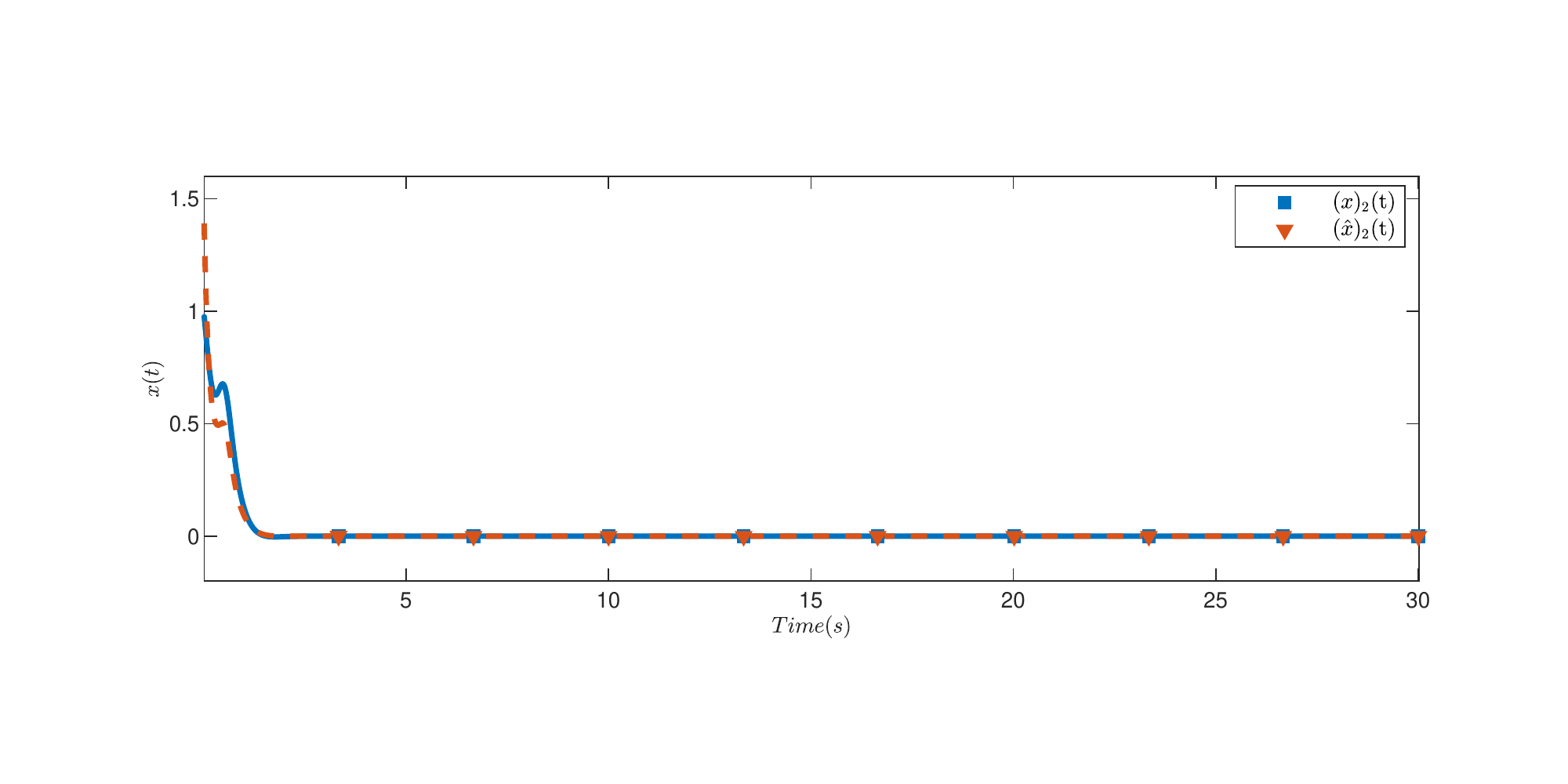}
  \vspace{-15mm}
		\caption{The trajectories of the actual state $x_2$ and estimated state  $\hat{x}_2$.}
		\label{fig:state2_sim1}
\end{figure}
\begin{figure}
  \vspace{-3mm}
         \centering
		\includegraphics[width=1\columnwidth]{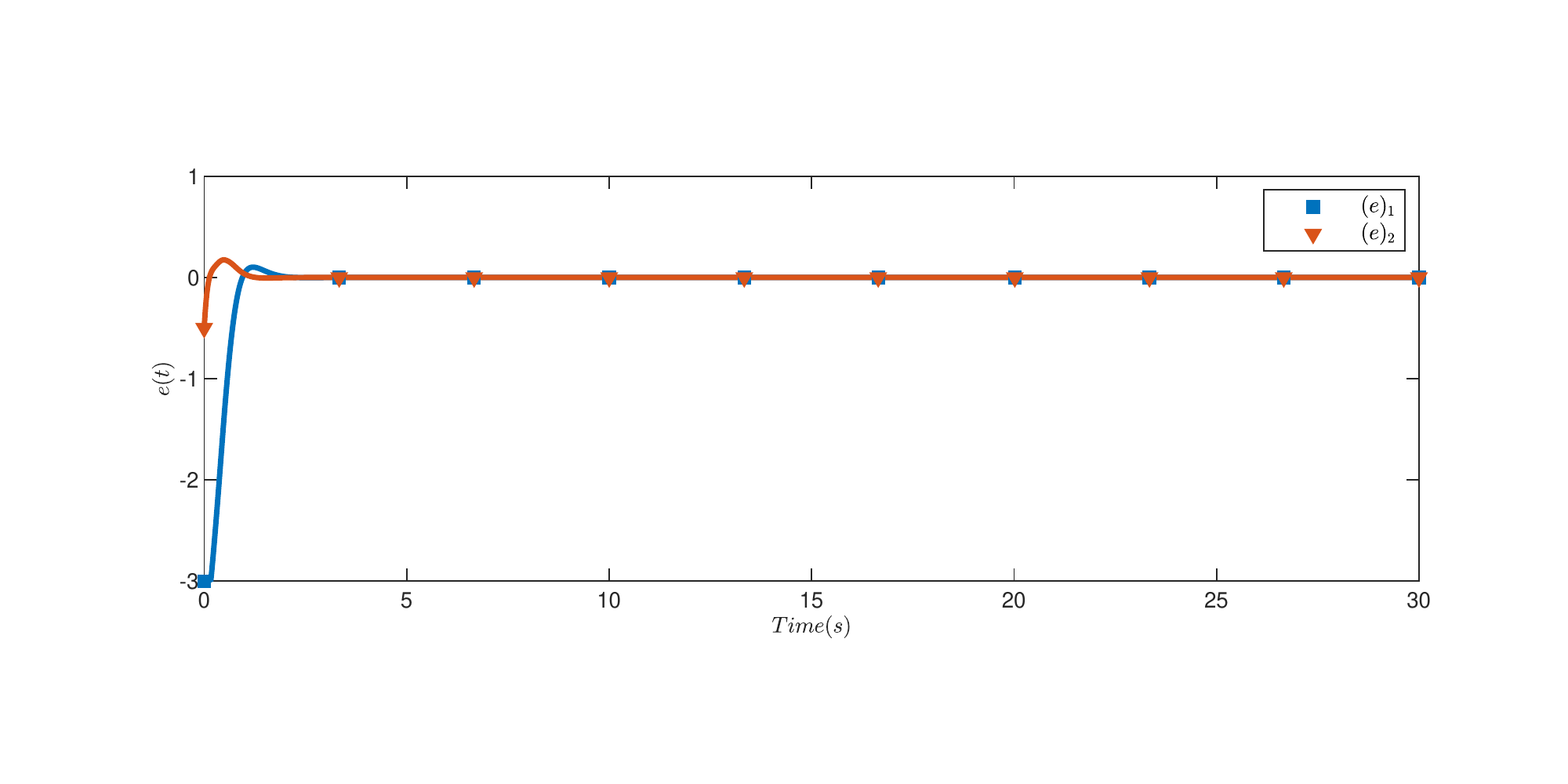}
             \vspace{-15mm}
            \caption{Estimation errors between the actual states and the estimated states}.
		\label{fig:error_sim1}
\end{figure}
\begin{figure}
  \vspace{-3mm}
        \centering		\includegraphics[width=1\columnwidth]{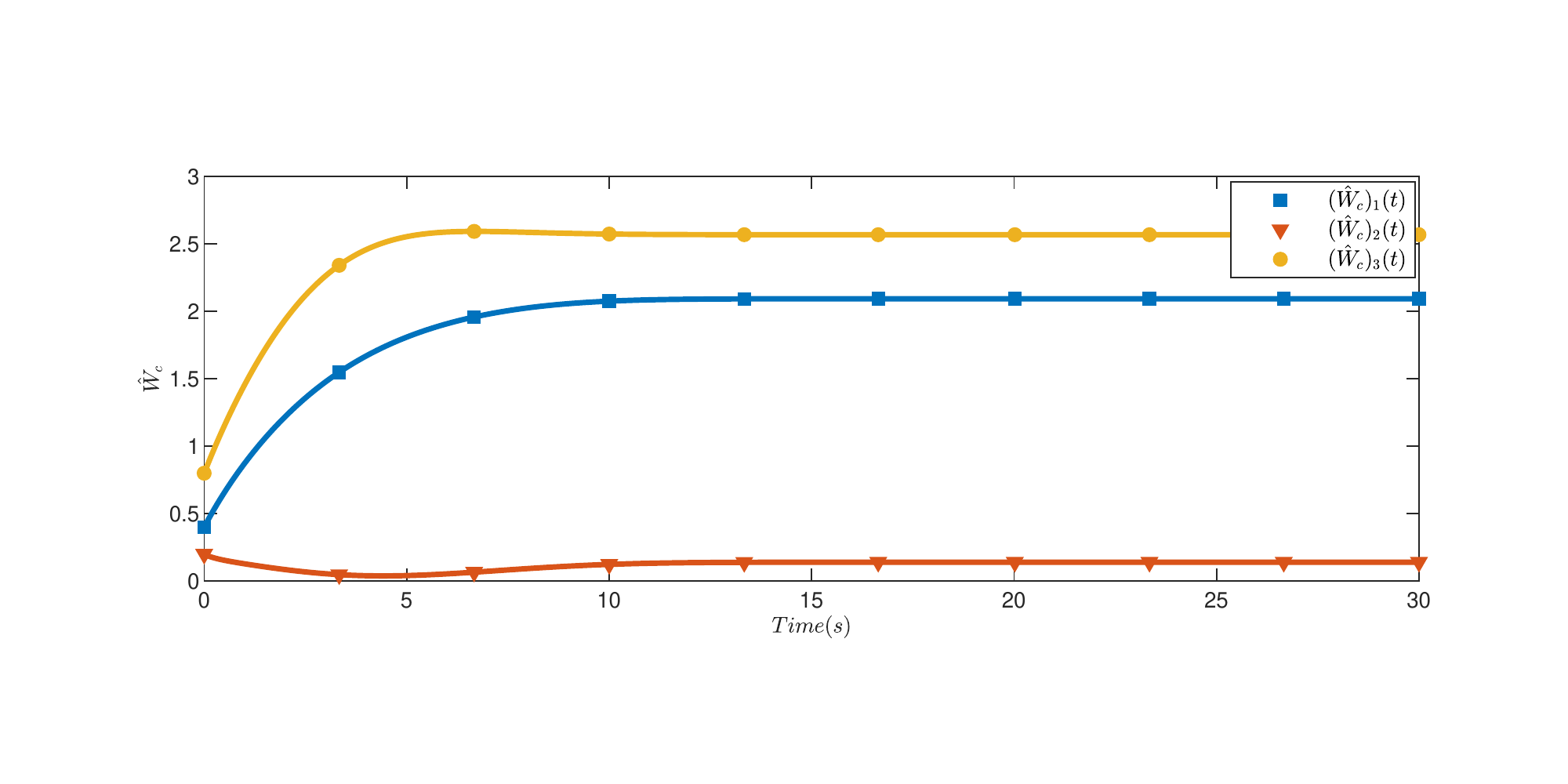}
  \vspace{-15mm}
		\caption{Estimated value function weights for the two-state dynamical system.}
		\label{fig:weights_sim1}
\end{figure}
\begin{figure}
  \vspace{-3mm}
        \centering
		\includegraphics[width=1\columnwidth]{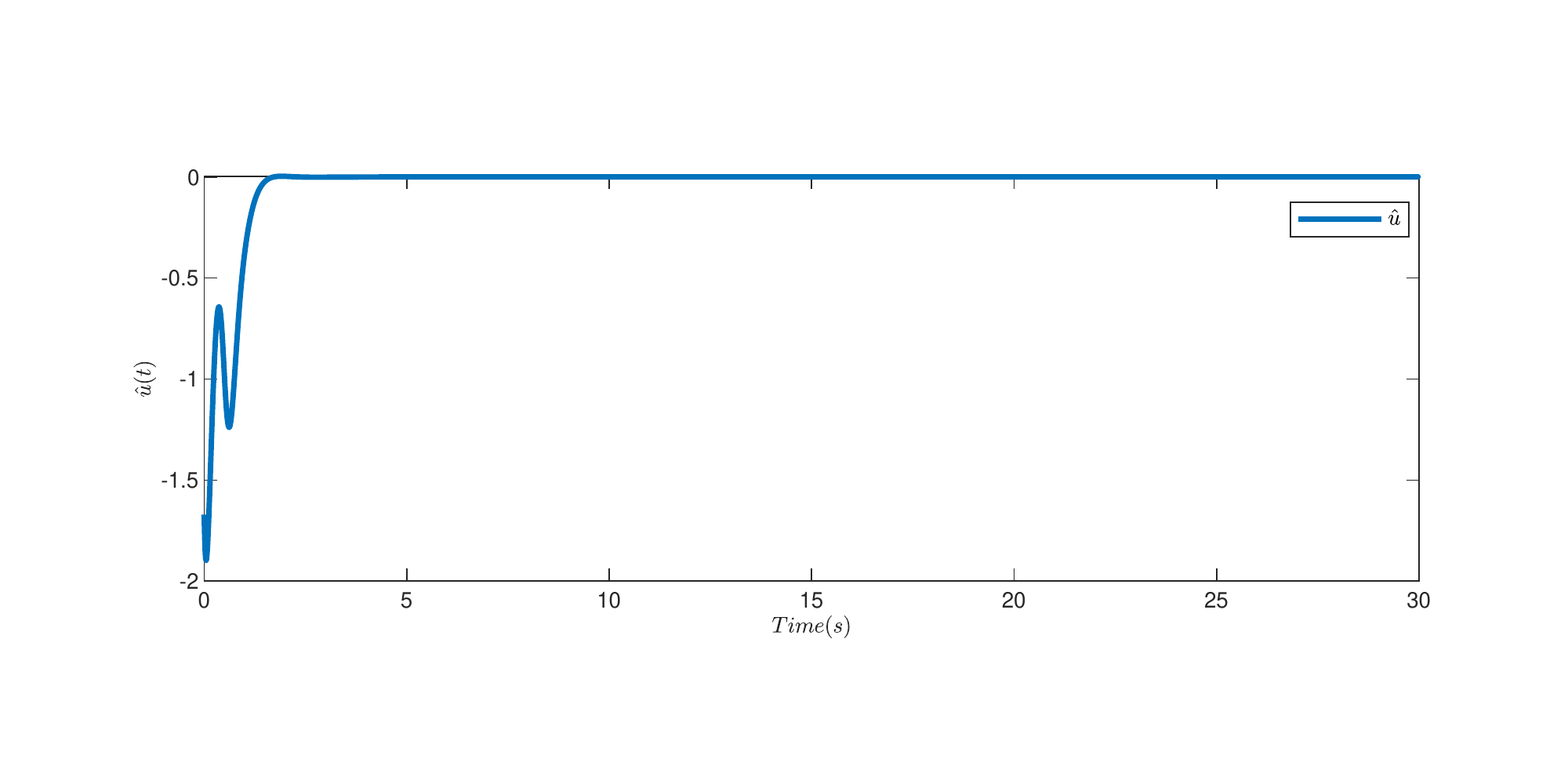}
            \vspace{-15mm}
		\caption{Trajectory of constrained control input.}
		\label{fig:control_sim1}
\end{figure}

\subsection{Results for the two-state system} \label{ Result_sim1}
Fig.\ref{fig:state1_sim1} and Fig.\ref{fig:state2_sim1} show that the trajectories of the actual states and state estimates converge to the same values, indicating the effectiveness of the designed observer in state estimation. From Fig.\ref{fig:error_sim1}, the resulting state estimation error is shown to converge to zero which demonstrates that the observer gains obtained by solving the linear
matrix inequalities (LMIs) satisfy the conditions required for stability developed in Section ~\ref{section:sectorFormulation}. From the simulation, the final values of $P$ and the observer gains $L$, $H$, and $K$ are obtained as,
\begin{gather*}
   P = \begin{bmatrix}
       0.75067  &  0.80202\\
    0.80202  &  1.9477
    \end{bmatrix}, \
    L = \begin{bmatrix}
   -18.8704 \\ 10.1087
    \end{bmatrix},\\
    H = \begin{bmatrix}
   1.1779 \\ -5.0084e-17
    \end{bmatrix}, \text{ and } 
    K = \begin{bmatrix}
    -5.0084e-17 \\ -5.0084e-17
    \end{bmatrix}
\end{gather*}The weight estimates are shown to be UUB in  Fig.\ref{fig:weights_sim1}. The controller in  Fig.\ref{fig:control_sim1} is shown to maintain stability while learning the value function and the observer ensures convergence of estimates of the states to their true values. The converged values match the values numerically obtained in \cite{SCC.Huang.Jiang2015}. 

\section{Conclusion}\label{section:conclusion}
 An observer-controller framework for output feedback RL in input-constrained nonlinear systems is developed. LMIs are formulated to obtain observer gain matrices and an MBRL-based controller is developed that maintains stability while finding an approximate solution to the optimal control problem. Simulation results demonstrate the effectiveness of the developed method and local uniform ultimately boundedness of the system states is guaranteed using a Lyapunov-based stability analysis. 

Though the observer provides accurate state estimates, nonlinear systems with time-varying Jacobian bounds can experience rank deficiency in certain regions of the state space which could potentially make the LMI in \eqref{lmi} infeasible at points in those regions. If the LMI is poorly conditioned, then the choice of LMI solver can affect the simulation results. To address these numerical issues, the current LMI architecture can be augmented with techniques such as \cite{SCC.Bengt.Richard.ea2012} which uses a delta operator formulation of the LMI.

Future research will also involve introducing a system identifier into the observer RL architecture that learns the system's dynamics for systems where the parameters of the system model are uncertain.

\small
\singlespacing
\bibliographystyle{IEEETrans.bst}
\bibliography{scc,sccmaster,scctemp}
 
\end{document}